\documentclass[a4paper,10pt,reqno]{amsart}

\usepackage{amsmath,amsfonts,amsthm,amssymb,dsfont,mathabx}
\usepackage[alphabetic]{amsrefs}
\usepackage[OT4]{fontenc}
\usepackage{enumerate}
\usepackage{mathrsfs,mathtools}
\usepackage{enumerate,enumitem, xspace}
\usepackage[pdftex]{graphicx}
\usepackage{color}
\usepackage{float}
\usepackage[center, font=small,labelfont=sc,labelsep=none]{caption}

\long\def\symbolfootnote[#1]#2{\begingroup
\def\thefootnote{\fnsymbol{footnote}}\footnote[#1]{#2}\endgroup}

\newtheorem{theorem}{Theorem}[section]
\newtheorem{lemma}[theorem]{Lemma}

\newtheorem{prop}[theorem]{Proposition}

\newtheoremstyle{named}{}{}{\itshape}{}{\bfseries}{.}{.5em}{\thmnote{#3 }#1}
\theoremstyle{named}
\newtheorem*{namedtheorem}{Theorem}
\newtheorem*{namedlemma}{Lemma}
\theoremstyle{plain}

\newtheorem{mainlemma}{Lemma}

\theoremstyle{definition}

\newcommand{\executeiffilenewer}[3]{%
\ifnum\pdfstrcmp{\pdffilemoddate{#1}}%
{\pdffilemoddate{#2}}>0%
{\immediate\write18{#3}}\fi%
}

\begin{document}

\title{Acute Triangulation of constant curvature polygonal complexes}

\author[F.~Brunck]{Florestan Brunck}
%$^{\dag}$
\address{Institute of Science and Technology, Austria\\
Am Campus 1, 3400, Klosterneurburg, Lower Austria, Austria}
\email{florestan.brunck@ist.ac.at}
%\thanks{$\dag$ Supported by a SURA award from donours W.Manson \& D.Catterson and funding from the Natural Sciences and Engineering Research Council (NSERC)}

\maketitle

\begin{abstract}
\noindent
We prove that every 2-dimensional polygonal complex, where each polygon is given a constant curvature metric and belongs to one of finitely many isometry classes can be triangulated using only acute simplices. There is no requirement on the complex to be finite or even locally finite. 
%Our method is constructive and gives a simple and explicit algorithm to output the refining acute triangulations but does not produce explicit angle bounds. 
\end{abstract}

\section{Introduction}
\label{intro}

In this section, we provide some brief historical background to situate our results, introduce some preliminary definitions and notations, and state our main result.

Previous existence results of acute refining triangulations focused heavily on the Euclidean setting, where results have been obtained in full generality and with increasingly detailed analysis (see for example the original proof of Buraglo and Zalgaler \cite{BZ} or more recently Saraf for Euclidean polyhedral surfaces \cite{SA}). While there has been a lot of recent developments in Euclidean acute triangulations and meshing algorithms (\cite{Zam}), very little is known in the non-Euclidean setting. To the best of our knowledge, there currently exists only two results dealing in some generality with acute triangulations in a non-Euclidean setting. The first characterises combinatorially all the acute triangulations of the sphere (\cite{KW}), and the second (more closely related to ours) establishes an existence result for Riemannian surfaces, with the caveat of only offering a highly non constructive approach (\cite{CDV}). We also mention in passing some results which focused on establishing minimal acute triangulations on non planar flat surfaces (see for example \cite{Yuan-Zam} and \cite{Itoh-Yuan}), but to our knowledge, the only instance which does not deal with a flat metric is the very specific case of spherical triangles (\cite{IZ}). 

In this present work, we provide an existence result for the general class of constant curvature polygonal complexes, as a stepping stone to tackle the most general 2-dimensional case of Riemannian polygonal complexes. Our proof is constructive and provides a simple and explicit algorithm to output the acute triangulation, providing a tractable link between known methods in the Euclidean case and the new general constant curvature case.
%doing so, we spell out the first constructive algorithm to create acute refining triangulations in a Non-Euclidean setting.

In our setting, all geodesics will be taken to be minimal. Throughout the article, $M_{\kappa}^{2}$ will stand for the complete, simply-connected Riemannian 2-manifold of constant sectional curvature $\kappa$. A \textit{polygon} $P$ in $M_{\kappa}^{2}$ is defined as a compact subset of $M_{\kappa}^{2}$ bounded by a polygonal Jordan curve consisting of finitely many geodesic segments. We point out to the reader that, according to our definition in the $\kappa>0$ case, a polygonal Jordan curve defines exactly two polygons. Note also that we do not require the geodesics forming the polygonal Jordan curve to be unique geodesics, and such a polygon may very well form a non-convex set. A \textit{constant curvature polygonal complex} is a 2-dimensional cell complex (possibly infinite and not locally finite) where each cell is taken to be one of finitely many different polygons $P_1,P_2,\ldots,P_N$, each respectively lying on the surface of constant curvature $M_{\kappa_1}^2,M_{\kappa_2}^2,\ldots,M_{\kappa_N}^2$. 

Our main result is the following:

\begin{namedtheorem}[Main]
\label{thm:main}
For any constant curvature polygonal complex $C$, there exists an acute refining triangulation of $C$. 
\end{namedtheorem}

A \textit{constant curvature triangle complex} is a 2-dimensional cell complex where each cell is taken to be a \textit{geodesic triangle}. The meaning of ``geodesic triangle'' here is taken to be slightly more restrictive than that of a 3-sided polygon. A \textit{geodesic triangle} $T$ in the surface $M^2_\kappa$ of constant curvature $\kappa$ is defined as a triple of points of $M^2_\kappa$, together with a choice of three geodesic segments joining each pair of points. If $\kappa \leq 0$, $M^2_\kappa$ is uniquely geodesic and any triple of points defines a unique geodesic triangle. If $\kappa > 0$ however, there exists a unique geodesic between two points if and only if the distance between them is strictly less than $\frac{\pi}{\sqrt{\kappa}}$ (\cite{BH}). For our future constructions to be well-defined in the positive curvature case, we then require the three vertices of a triangle to lie in the same open hemisphere (the largest uniquely geodesic convex set in $M_\kappa^2$, see \cite{BH}). Equivalently, we could require the perimeter of our triangles to be strictly less than $\frac{2\pi}{\sqrt{\kappa}}$. In the positive curvature setting, we shall then understand the meaning of ``geodesic triangle'' to include these restrictions on the possible triples of points. 

We remark that a polygon in $M_\kappa^2$ can always be triangulated using only geodesic triangles, so that there is no loss of generality in restricting our attention to constant curvature triangle complexes. Indeed, the original proof of Dehn that every Euclidean polygon can be triangulated by diagonals carries to any non-Euclidean geometry which satisfies Hilbert's axioms of incidence and order (see \cite{Poincare} for Poincar\'e's account of Hilbert's axioms, and \cite{guggenheimer1} and \cite{guggenheimer2} for the original proof of Dehn and its adaptation in the axiomatic geometry setting). In the case where $\kappa<0$, the surface $M^2_\kappa$ already satisfies Hilbert's axioms of incidence and order and the result is immediate. In the $\kappa>0$ case, we need only restrict ourselves to an open hemisphere to satisfy all the axioms (only Hilbert's second incidence axiom is violated on the sphere). Since the polygonal Jordan curve bounding a polygon consists of finitely many geodesic segments, there exists a choice of great circle which does not prolong any of its geodesic sides. Such a great circle cuts the original polygon into finitely many polygons which are each contained in an open hemisphere and can therefore be triangulated by diagonals. Each of these triangulations uses only diagonals and thus their union yields a triangulation of the initial polygon using only geodesic triangles.

% Since every 2-dimensional $C^2$ manifold with boundary can be triangulated with triangles which are all arbitrarily and uniformly small (see section 10.2 of \cite{Mathijs}, which corrects an earlier proof of a theorem of Dyer and Ramsay in \cite{Dyer}), there is no loss of generality in restricting our attention to constant curvature triangle complexes. We point out to the reader that the previously cited articles do not address the case of a manifold with boundary, which is still open in arbitrary dimension. However, a known argument of Whitney gives an elementary reduction to the case of a manifold without boundary in the 2-dimensional case. The resolution of the higher dimensional cases is treated by upcoming work of Wintraecken \& Al \cite{Wintra}.

\section{Overview of the proof}
\label{overview}

Consider a constant curvature triangle complex $T$. Our proof follows three major steps, which we summarise here informally for the reader.

\begin{enumerate}
\item We introduce a preliminary elementary sequence of refining triangulations $T_0, T_1, T_2 \ldots $ of the triangle complex $T\eqqcolon T_0$, such that $T_{i+1}$ is a refinement of $T_i$. Previous work (\cite{Flo}) allows us to rely on two key properties of this sequence of refining triangulations:
\begin{enumerate}
    \item \textit{Triangles get arbitrarily small.} For all $\epsilon>0$, there exist a subdivision step $N$ such that all the triangles of $T_N$ lie inside a ball of radius $\epsilon$ (Lemma A).
    %moreover, the side lengths of all the triangles decrease exponentially with respect to the subdivision steps 
    \item \textit{Triangle angles are uniformly bounded.} There exists $\delta>0$ such that, for all $n\in \mathbb N$, all the angles of $T_n$ lie in the interval $(\delta, \pi-\delta)$ (Lemma B).
\end{enumerate}
\item We fix an appropriately chosen subdivision step $N$ and associate to $T_N$ a Euclidean triangle complex $\overline{T_N}$, called its \textit{Euclidean comparison complex}. We then apply a new result of Chris Bishop (\cite{Bishop}) to produce an acute refining triangulation of this Euclidean complex which preserves a (possibly different) uniform angle bound on its triangle angles (Bishop's Theorem) and provides a uniform bound to separate interior vertices from the edges of the complex $\overline{T_N}$ it refines (Bishop's Lemma).
\item We construct a family of diffeomorphisms to pull back the previously obtained Euclidean acute triangulation of $\overline{T_N}$ onto the original constant curvature triangle complex $T$, working triangle by triangle in a way that ensures that all angles stay acute in the process (Lemma D). The resulting subdivision is an acute dissection of $T$ (Proposition \ref{dissection}), but not necessarily a refining triangulation.
\item We displace vertices lying in the interior of the edges of the triangles of $T_N$ to transform the previous dissection into a refining triangulation of $T$. We do so while guaranteeing that the displacement needed is small enough to preserve all the acute angles.
\end{enumerate}

\section{The Iterated Medial Subdivision}
\label{subdivision}
In this section, we concern ourselves with step (1) and define a natural and elementary inductive subdivision scheme of a constant curvature triangle complex. All the results and constructions used in this section are presented in previous work (see \cite{Flo}), which we recall here as needed. 

In \cite{Flo}, the \textit{iterated medial triangle subdivision} (see Figure \ref{fig:medial-subdivision}) of a constant curvature triangle complex $T$ is defined as the following sequence $T_0, T_1, T_2, \ldots$ of refining triangulations:
\begin{itemize}
    \item $T_0=T$
    \item $T_{n+1}$ is obtained from $T_n$ by adding the midpoints of the edges of $T_n$ and, within each triangle of $T_n$, pairwise connecting its 3 midpoints by geodesic segments  (this creates $4$ new sub-triangles for each triangle of $T_n$). 
\end{itemize}

\begin{figure}[H]\centering
  \includegraphics[page=1]{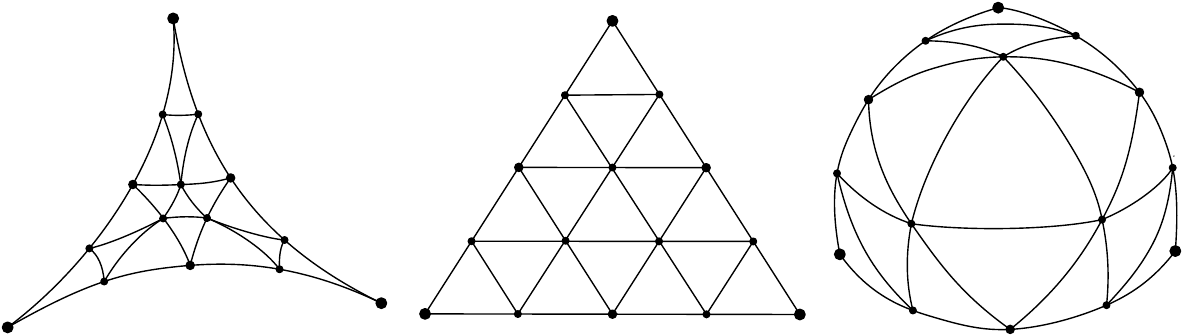}
  \caption{: The first two medial triangle subdivisions of a triangle in $\mathbb H^2$, $\mathbb E^2$ and $\mathbb S^2$.}
  \label{fig:medial-subdivision}
\end{figure}

The following two lemmas, which we shall need later on, are immediate consequence of Theorems A and C in (\cite{Flo}).

\begin{mainlemma}
\label{lemma:A}
For any constant curvature triangle complex $T$, and for all $\epsilon>0$, there exists $N\in \mathbb N$ such that, for all $n>N$, all the edge lengths of $T_n$ are smaller than $\epsilon$.
\end{mainlemma}

\begin{mainlemma}
\label{lemma:B}
For any constant curvature triangle complex $T$, there exists $\delta>0$ such that, for all $n\in\mathbb N$, all the angles of $T_n$ lie in the interval $(\delta, \pi -\delta)$.
\end{mainlemma}

\section{Acute Triangulation of the Euclidean Comparison Complex}
\label{sec2}

We start by recalling that, given a metric space $\mathcal M_\kappa$ of constant curvature $\kappa$ and a geodesic triangle in $\mathcal M_\kappa$ with vertices $P,Q,R$, there exists a unique \textit{comparison triangle} in $\mathbb E^2$ with vertices $\overline{P},\overline{Q}, \overline{R}$ such that $d_{\,\mathbb E^2}(\overline{P}, \overline{Q}) = d_X(P,Q)$, $d_{\,\mathbb E^2}(\overline{P}, \overline{R}) = d_X(P,R)$ and $d_{\,\mathbb E^2}(\overline{Q}, \overline{R}) = d_X(Q,R)$ (\cite{BH}, 1.10 and 2.14). Given a constant curvature triangle complex $T$, the \textit{Euclidean comparison complex} of $T$, denoted by $\overline{T}$, is defined as the simplicial complex obtained by gluing together the Euclidean comparison triangles of each triangle of $T$ according to their respective gluings in $T$. 

We shall be looking more specifically at the sequence of comparison complexes $\overline{T_0},\overline{T_1}, \overline{T_2},\ldots $ associated with the iterated medial triangle subdivision $T_0, T_1, T_2, \ldots$ of $T$. Note that $\overline{T_n}$ is \textit{not} a refinement of $\overline{T_{n-1}}$ (see Figure \ref{fig:comparison-complex}). Instead, the Euclidean complexes $\overline{T_0},\overline{T_1}, \overline{T_2}\ldots $ should be thought of as increasingly accurate piecewise linear approximations of $T$.
%now explain how to construct a Euclidean triangle complex for each successive step of the iterated medial triangle subdivision of the original Riemannian triangle complex. The resulting family of Euclidean complexes can be seen as approximating the original triangle complex as more steps of the iterated triangle subdivision are taken. 

% Our next goal is to construct a Euclidean triangle complex from the original Riemannian triangle complex in such a way that we  to pull an acute refinement of this Euclidean complex back to the Riemannian complex and obtain an acute refinement of the original Riemannian triangulation. 
% Note also that the Euclidean comparison complex is not guaranteed to be CAT$(0)$, even if the original Riemannian complex is (for example, an interior vertex in the Euclidean comparison complex will have an angle lesser or greater than $2\pi$ surrounding it depending on whether the vertex was in the interior of a positively or negatively curved triangle). 

\begin{figure}[H]\centering
  \includegraphics[page=2]{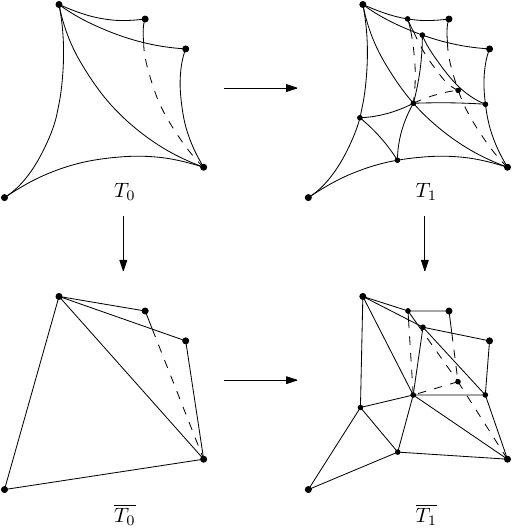}
  \caption{}
  \label{fig:comparison-complex}
\end{figure}

In \cite{Bishop}, Bishop proves the following property for Euclidean simplicial complexes:

\begin{namedtheorem}[Bishop's]
\label{bishop}
Let $\delta>0$ and $C$ be a Euclidean simplicial complex with angles lying in the interval $(\delta, \pi-\delta)$. There exists an acute refinement $B(C)$ of $C$ with angles in the interval $(\theta, \frac{\pi}{2}-\theta)$, where $\theta>0$ depends only on $\delta$. 
\end{namedtheorem}

We first observe that Lemma $B$ extends to the comparison complexes of the iterated medial subdivisions of a constant curvature triangle complex with uniform angle bounds (Lemma 4.1). From there, Bishop's Theorem gives us acute triangulations of the comparison complexes of all the iterated medial subdivisions of $T$ with a uniform angle bound (Lemma C). 

\begin{lemma}
Consider a constant curvature triangle complex $T$ and a constant $\delta >0$ such that, for all $n\in\mathbb N$, all the angles of $T_n$ lie in $(\delta, \pi-\delta)$. Then all the angles of $\overline{T_n}$ also lie in the interval $(\delta, \pi - \delta)$.
\end{lemma}

\begin{proof}
In the hyperbolic case (resp. the spherical case), the Euclidean comparison angles are strictly larger (resp. smaller) than the original angles (\cite{BH}). We thus obtain a lower bound of $\delta$ (resp. an upper bound of $\pi-\delta$) through Lemma B. Since angles must add up to $\pi$ in the Euclidean comparison triangles, this also translates to an upper bound of $\pi-2\delta$ (resp. a lower bound of $2\delta$).
\end{proof}

\begin{mainlemma}
For any constant curvature triangle complex $T$, there exists $\Delta>0$ such that, for all $n\in\mathbb N$,  $B(\overline{T_n})$ is an acute triangulation of $\overline{T_n}$ with all triangle angles lying in the interval $(\Delta, \frac{\pi}{2}-\Delta)$. 
\end{mainlemma}
\begin{proof}
Lemma B guarantees the existence of $\delta>0$ such that, for all $n\in\mathbb N$, all the angles of $T_n$ lie in $(\delta, \pi-\delta)$. Lemma 4.1 then transfers this bound to $\overline{T_n}$, for all $n\in\mathbb N$. Finally,  for all $n\in\mathbb N$, Bishop's Theorem guarantees that there exists $\Delta_n>0$ such that $B(\overline{T_n})$ is an acute refinement of $\overline{T_n}$ with all triangle angles lying in the interval $(\Delta_n, \frac{\pi}{2}-\Delta_n)$, where the various values of $\Delta_n$ can be taken to be identical since $\Delta_n$ depends only on the $\delta$ obtained from Lemma B.
\end{proof}

\section{From Euclidean to Spherical and Hyperbolic}
\label{sec3}

All metrics of constant curvature have the advantage of being locally projectively equivalent, i.e. there exist local diffeomorphisms between any two small enough neighbourhoods of any two spaces of constant curvature which send local geodesics to local geodesics (this was known already to Lagrange in positive curvature \cite{Lagrange}, and generalised to negative curvature by Beltrami \cite{Beltrami1},\cite{Beltrami2}). Such local diffeomorphisms are called \textit{geodesic maps}. 

In both the spherical case and the hyperbolic case, the geodesic maps providing the local projective equivalence (or a global one in the hyperbolic case) are given by well chosen projections in $\mathbb R^3$ of either the open half-sphere or of a single sheet of the standard two-sheeted hyperboloid. We briefly recall the definitions of these diffeomorphisms, called respectively the \textit{gnomonic projection} in the spherical setting and the \textit{Klein model} in the hyperbolic setting. 

\begin{figure}[H]
	\begin{center}     
    	\includegraphics[width=\textwidth]{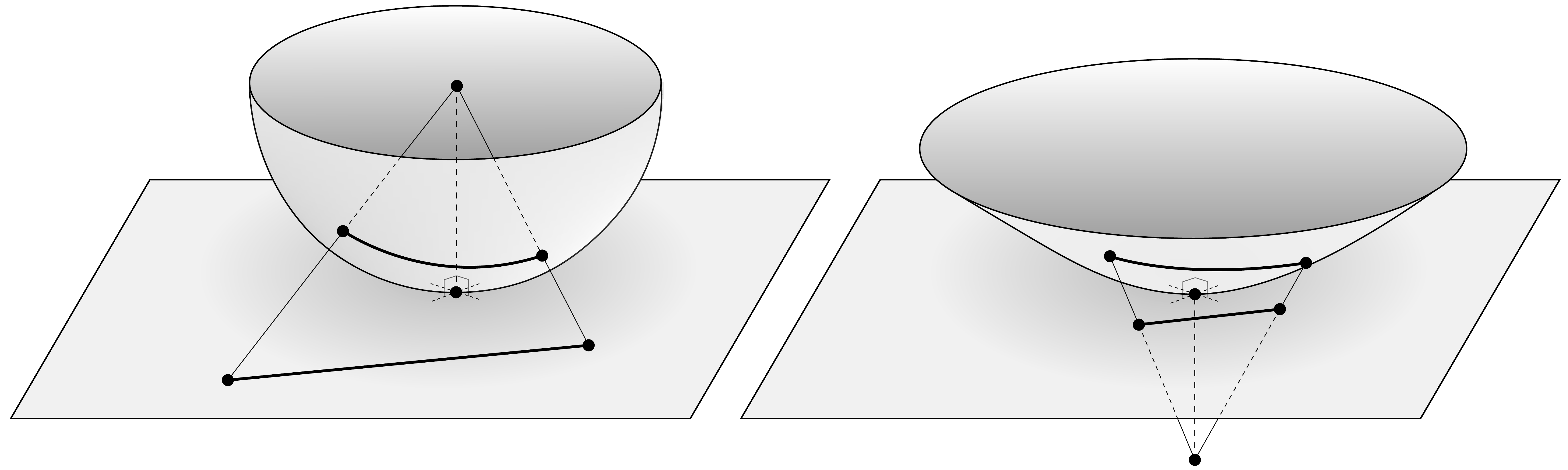}
        	\caption{}
        	\label{fig:gnomonic}
	\end{center}
\end{figure}

Consider the unit sphere in $\mathbb R^3$ and a choice of great equatorial circle. The gnomonic projection is defined as the radial projection of one of the associated open hemisphere from the centre of the sphere in $\mathbb R^3$ onto the plane tangent to the pole contained in the chosen hemisphere. Likewise, the Klein model map is obtained by projecting a single sheet of the standard two-sheeted hyperboloid from its centre onto the plane tangent to the apex of the chosen single sheet (see Fig. \ref{fig:gnomonic}). Since the gnomonic map and the Klein model map play symmetrical roles in positive and negative constant curvature, we will use the same letter $\pi$ to denote either one and rely on the spherical or hyperbolic context to clarify the distinction between them. For the same reason, we will denote by $P$ both the selected spherical pole used in the gnomonic projection and the apex of the chosen hyperboloid sheet used in the Klein model. We denote by $B(P,\epsilon)$ the ball of radius $\epsilon$ centred at $P$ in both the open hemisphere and the hyperboloid sheet. Note that this construction is readily adapted for a non-unit curvature $\kappa$ by adjusting the curvature of the sphere (resp. hyperboloid sheet) under consideration.

Given a geodesic triangle $t$ in $\mathcal M^{2}_\kappa$, our aim for this section is to construct a diffeomorphism $\phi_t$ between $t$ and its Euclidean comparison triangle $\overline{t}$ in such a way that we may guarantee the following lemma:

\begin{mainlemma}
For all $\iota>0$, there exists $\epsilon>0$ such that, for any geodesic triangle $t$ contained in a ball of radius $\epsilon$, there exists a diffeomorphism $\phi_t: t\rightarrow \overline{t}$ mapping geodesics to geodesics and such that, for any 3 points $A,B,C\in t$, we have $\angle \phi(A)\phi(B)\phi(C)\in (\angle ABC -\iota, \angle ABC + \iota)$. 
\end{mainlemma}

%giving an overview of our method to pull back the restriction of the acute triangulation obtained in section 4 to a single triangle $\bar{t}$ of $\overline{T}$ onto its original triangle $t$ in $T$.
\begin{proof}
In both the spherical and the hyperbolic cases, we begin by translating $t$ so that the circumcenter of $t$ coincides with $P$. There is a unique translation $\tau_t$ that achieves this and it does not modify angles. The following claim guarantees that the gnomonic projection map (resp. the Klein model map) do not affect the angles much for triangles that stay close to the pole (resp. apex).

\textbf{Claim 1.} For all $\delta>0$, there exists $\epsilon >0$, such that, for any 3 points $A,B,C\in B(P,\epsilon)$, $\angle \pi(A)\pi(B)\pi(C)\in (\angle ABC -\delta, \angle ABC + \delta)$. 

\begin{proof}[Proof of Claim 1.]
Let us fix $\delta>0$ and denote by $f$ and $g$ the two unique geodesics passing respectively through $A$ and $B$, and $B$ and $C$, with tangent vectors $u$ and $v$ at $B$ (for any small enough $\epsilon$, we can assume the geodesics to be unique). Note that $d\pi_P=\text{Id}_{\mathbb R^2}$ (identifying the tangent plane at $P$ with $\mathbb R^2$), and we can thus choose a small enough $\epsilon$ such that, if $A,B,C$ are $\epsilon$-close to $P$,  the angle between $d\pi(u)$ and $d\pi(v)$ differs less than $\delta$ from the angle between $u$ and $v$ (i.e. $\angle pqr$). The fact that $\pi$ sends geodesics to geodesics guarantees that the angle between $d\pi(u)$ and $d\pi(v)$ is exactly $\angle \pi(p)\pi(q)\pi(r)$. 
\end{proof}

To build our desired diffeomorphism $\phi_t$, we compose $\pi\circ \tau_t$ by the unique affine map $\psi_t$ taking $\pi(t)$ to $\overline{t}$, i.e. we write $\phi_t$ to denote $\psi_t\circ\pi\circ \tau_t$. Note first that both translations and affine maps are geodesics maps, and since $\pi$ is geodesic as well, so is $\phi_t$. For two points $p,q\in \mathbb E^2$, let us denote by $|pq|$ the Euclidean distance between them. The following claim guarantees that, for a small enough triangle $t$, the image of $t$ under $\pi$ is almost $\overline{t}$ (and thus $\psi$ does not distort angles too much either). 

\textbf{Claim 2.} Given points $p,q\in B(P,\epsilon)$, the ratio $d_{M^2_\kappa}(p,q)/|\pi(p) \pi(q)|$ approaches 1 as $\epsilon$ becomes smaller (from above if $\kappa<0$, from below if $\kappa>0$).

\begin{proof}[Proof of Claim 2.]
Recall that, for any $\kappa\neq 0$, the metric tensor for the associated Klein model or gnomonic projection at the point $X$ with Euclidean coordinates $(x_1, x_2)$ is given by (\cite{HG}):

\[
\boldsymbol{g}_{ij}(X) = \frac{\delta_{ij}}{1+\kappa(x_1^2+x_2^2)} + \frac{x_1 x_2}{(1+\kappa(x_1^2+x_2^2))^2}
\]

Let us denote by $\boldsymbol{g}^{\mathbb E^2}_{ij}(X)$ the standard Euclidean metric tensor at $X$. Then, for a point $X\in B(P,\epsilon)$, and when $\kappa<0$, we obtain:

\[
1 \leq \frac{\boldsymbol{g}_{ij}(X)}{\boldsymbol{g}^{\mathbb E^2}_{ij}(X)} \leq \frac{1}{1+\kappa\epsilon^2}+\frac{\epsilon^2}{(1+\kappa\epsilon^2)^2} \underset{\epsilon \rightarrow 0} \longrightarrow 1
\]

Likewise, when $\kappa >0$, we obtain:

\[
1 \underset{0\leftarrow \epsilon}\longleftarrow \frac{1}{1+\kappa\epsilon^2}+\frac{\epsilon^2}{(1+\kappa\epsilon^2)^2} \leq \frac{\boldsymbol{g}_{ij}(X)}{\boldsymbol{g}^{\mathbb E^2}_{ij}(X)} \leq 1
\]

\end{proof}

This claim implies that $\psi_t$ can be chosen arbitrarily close to a translation and thus proves the following claim:

\textbf{Claim 3.}  For all $\delta>0$, there exists $\epsilon >0$, such that, for any 3 points $A,B,C\in \pi(t)$, $\angle \psi_t(A)\psi_t(B)\psi_t(C)\in (\angle ABC -\delta, \angle ABC + \delta)$. 

Using Claim 1 and Claim 3 with $\delta = \frac{\iota}{2}$ finishes the proof of the lemma.

\end{proof}

\section{Pulling Back the Euclidean Acute Triangulation}
\label{sec4}

Consider a constant curvature triangle complex $T$ and a family  $\Phi_T=\{\phi_t:t\rightarrow \overline{t}\>|\> t\in T\}$ of geodesic diffeomorphisms between each triangle of $T$ and their Euclidean comparison triangles. Given a refining triangulation $\overline{T}'$ of $\overline{T}$, we obtain a dissection of $T$ by pulling back, for each triangle $t$ in $T$, the restriction of $\overline{T}'$ to $\overline{t}$ using $\phi^{-1}_t$. We shall denote this dissection by $\Phi_T^{-1}(\overline{T}')$ (see Fig. \ref{fig:heuristics}). Together with Lemma C and Lemma D, this process guarantees the existence of an acute dissection of any constant curvature triangle complex $T$ with finitely many isometry types of triangles:
%Additionally, in this context, every vertex $\overline{V}$ (resp. edge $\overline{e}$) of $\overline{T'}$ has a unique preimage in $\Phi_T^{-1}(\overline{T}')$ which we will denote by removing the overhead bar, i.e $V$ (resp. $e$).
\begin{prop}
\label{dissection}
For any constant curvature triangle complex $T$, there exist $\delta>0$ and $N\in\mathbb N$ such that, for all $n>N$, $\Phi^{-1}_{T_n}(B(\overline{T_n}))$ is an acute dissection of $T$ with angles lying in the interval $(\delta, \frac{\pi}{2}-\delta)$.
\end{prop}
\begin{figure}[H]\centering
  \includegraphics[page=3]{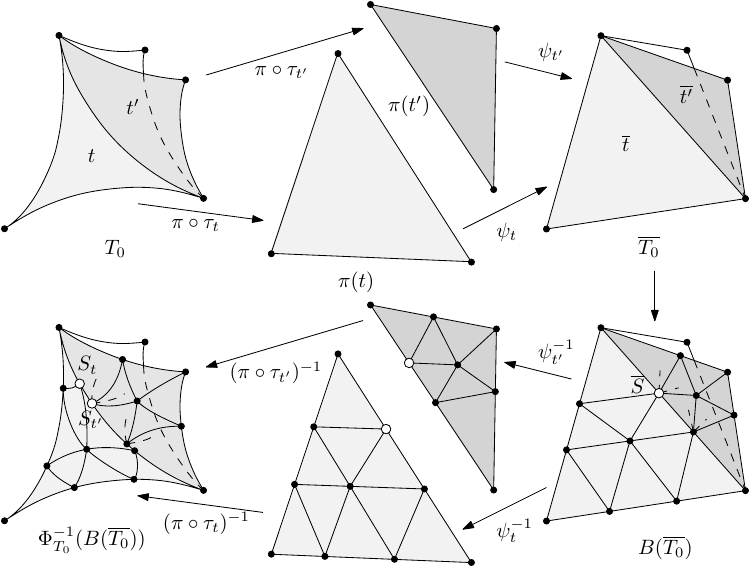}
  \caption{}
  \label{fig:heuristics}
\end{figure}
\begin{proof}
Choose $\Delta$ as in Lemma C and let us pick $\delta=\frac{\Delta}{2}$. Using Lemma D, we can select $\epsilon>0$ such that, for any geodesic triangle $t$ lying inside a ball of radius $\epsilon$, there is a diffeomorphism $\phi_t: t\rightarrow \overline{t}$ which perturbs the angle between any three points of $t$ by strictly less than $\frac{\Delta}{2}$. By Lemma $A$, there exists a subdivision step $N\in\mathbb N$ such that, for all $n>N$, all edges of $T_n$ are less than $\epsilon$, so that every triangle of $T_n$ lies in a ball of radius $\epsilon$. Fixing $n>N$, let us denote by $\Phi_{T_n}$ the collection of all such geodesic diffeomorphisms $\phi_t$, for each triangle $t$ in $T_n$.  By construction, $\Phi^{-1}_{T_n}(B(\overline{T_n}))$ is a dissection of $T$. Lemma D guarantees that all angles of $\Phi^{-1}_{T_n}(B(\overline{T_n}))$ lie in the interval $(\delta, \frac{\pi}{2}-\delta)$. 
\end{proof}

\textbf{Remark.} In the previous proof, for any given triangle $t$, $\phi_t$ is not an isometry on the sides of $t$. Because of this, given a point $\overline{S}$ of $B(\overline{T_N})$ which lies in the interior of an edge $\overline{e}$ of $\overline{T_N}$ shared by two triangles $\overline{t}$ and $\overline{t'}$, it is likely that $\phi_t^{-1}(S)\eqqcolon S_t$ and $\phi_{t'}^{-1}(S)\eqqcolon S_{t'}$ are two distinct points on the edge $e\coloneqq \phi_t^{-1}(\overline{e})=\phi_{t'}^{-1}(\overline{e})$ of $T_N$ (see Fig. \ref{fig:heuristics}). Therefore we cannot immediately guarantee that $\Phi^{-1}_{T_N}(B(\overline{T_N}))$ is more than just a dissection. 

% \begin{figure}[H]\centering
%   \includegraphics[page=3]{Acute-Triangulations.pdf}
%   \caption{}
%   \label{fig:heuristics}
% \end{figure}

We begin by providing the reader with a brief and informal summary of how we intend to resolve this problem. Our approach is to ``snap back together'' non-matching edge points from different adjacent triangles to resolve the dissection into a triangulation. In order to do so, we combine two main ingredients. First, an additional lemma of Bishop (see \cite{Bishop}) gives us a uniform bound for separating the interior points of our refining acute triangulations away from the original triangle edges of the complex it refines (Bishop's Lemma). We will show that this uniform separation bound can be brought back to the non-Euclidean complex (Lemma E). Second, claim 2 of Lemma D guarantees that for a large enough subdivision level $N$, each geodesic diffeomorphism $\phi_t^{-1}$ can be chosen arbitrarily close to an isometry. Therefore, non-matching associated edges points, each coming from the pullbacks of the restriction of the acute triangulation $B(\overline{T_N})$ to adjacent triangles of $T_N$, can be chosen to be arbitrarily close to each other (relative to the edge lengths of the triangle). We can then choose a large enough $N$ to guarantee that snapping back together corresponding edge vertices from two different triangles sharing an edge yields an arbitrarily small perturbation with respect to the uniform angle bound of $\Phi^{-1}_{T_N}(B(\overline{T_N}))$.

We begin by recalling the following lemma, proved by Chris Bishop in \cite{Bishop}:

\begin{namedlemma}[Bishop's]
Let $\delta>0$ and $C$ be a Euclidean simplicial complex with angles lying in the interval $(\delta, \pi-\delta)$. There exists a constant $\theta$, depending only on $\delta$, such that the acute refinement $B(C)$ of $C$ provided by Bishop's Theorem satisfies the following additional property $(\star)$: for each triangle $t$ in $C$, and each vertex $X$ of the restriction of $B(C)$ to the interior of $t$, all 6 angles formed by $X$ and the edges of $t$ are strictly greater than $\theta$.
\end{namedlemma}

\begin{figure}[H]\centering
  \includegraphics[page=4]{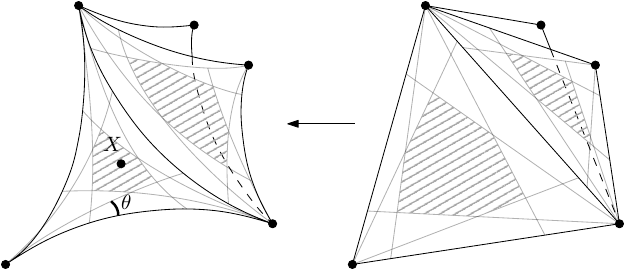}
  \caption{: Bishop's Lemma implies a separation bound in the non-Euclidean case as well: interior vertices must lie in the grey-hatched area. }
  \label{fig:separation}
\end{figure}

Combining Bishop's Lemma with Proposition \ref{dissection}, Lemma D and Lemma C gives us the following Lemma:

\begin{mainlemma}
For any constant curvature triangle complex $T$, there exists $\delta>0$, $N\in \mathbb N$ and $\theta>0$ (the latter depending only on $\delta$), such that, for all $n>N$:

\begin{enumerate}[label=(\roman*)]
  \item $\Phi^{-1}_{T_n}(B(\overline{T_n}))$ is an acute dissection of $T$ with all triangle angles lying in $(\delta, \frac{\pi}{2}-\delta)$.
  \item For each triangle $t$ in $T_n$, and each vertex $X$ of the restriction of $\Phi^{-1}_{T_n}(B(\overline{T_n}))$ to the interior of $t$, all 6 angles formed by $X$ and the edges of $t$ are strictly greater than $\theta$.
\end{enumerate}
\end{mainlemma}
\begin{proof}
Choose $\delta>0$ and $N\in\mathbb N$ according to Proposition \ref{dissection}, so that, for all $n>N$, $\Phi^{-1}_{T_n}(B(\overline{T_n}))$ is an acute dissection of $T$ with angles in $(\delta, \frac{\pi}{2}-\delta)$. Lemma C allows us to choose $\Delta>0$ such that, for all $n\in\mathbb N$, $B(\overline{T_n})$ has angles in $(\Delta, \frac{\pi}{2}-\Delta)$. We can thus apply Bishop's lemma to infer the existence of $\omega>0$, such that, for all $n\in\mathbb N$, $B(\overline{T_n})$ satisfies property $(\star)$ with angle $\omega$. Applying Lemma D with $\iota=\frac{\omega}{2}$, we know that there exists $\epsilon>0$ such that, for any triangle $t$ of $T_n$ contained in a ball of radius $\epsilon$, there exists a geodesic diffeomorphism $\phi_t$, modifying angles between any three points of $t$ by strictly less than $\iota$. By Lemma A, we can select an integer $M>N$ such that, for all $n>M$, all triangles of $T_{n}$ sit inside a ball of radius $\epsilon$. Fix such an integer $n$ and consider an interior vertex $X$ of a triangle $t$ of $T_n$. $X$ is mapped to a unique interior vertex $\phi_t(X)$ in $\overline{t}$, for which we know that the separation property $(\star)$ is verified with angle $\omega$. Lemma D then guarantees that the separation property is also verified for $X$ in $t$ with angle $\theta\coloneqq \omega-\iota=\frac{\omega}{2}$.
\end{proof}

\begin{proof}[Proof of the Main Theorem]
Recall that we need only prove the main theorem for a constant curvature triangle complex $T$. Given such a complex, Lemma E guarantees that we can choose $\delta$,  $N$ and $\theta$ such that properties $(i)$ and $(ii)$ hold. Fix an integer $n>N$ and, for each vertex $\overline{S}$ of $B(\overline{T_n})$ lying on an edge $\overline{P}\>\overline{Q}$ of $\overline{T_n}$ which is shared by (possibly infinitely many) triangles of type $\overline{t_1}, \overline{t_2}, \ldots, \overline{t_M}$, $M\geq 2$, merge all the vertices $\phi_{t_1}^{-1}(\overline{S}), \phi_{t_2}^{-1}(\overline{S}), \ldots, \phi_{t_M}^{-1}(\overline{S})$ of $\Phi^{-1}_{T_n}(B(\overline{T_n}))$ together into the unique point $S$ lying on the edge $PQ$ of $T_n$ such that $S$ is the image of $\overline{S}$ under the unique isometry from $\overline{P}\>\overline{Q}$ to $PQ$ which maps $\overline{P}$ to $P$. By construction, this process produces a refining triangulation $A_n$ of $T$ as we have now restored the exact combinatorial structure of $B(\overline{T_n})$ . Note that this process is still well-defined even in the event where there are infinitely many triangles meeting along $\overline{P}\>\overline{Q}$, in which case each of the $M$ points $\phi_{t_1}^{-1}(\overline{S}), \phi_{t_2}^{-1}(\overline{S}), \ldots, \phi_{t_M}^{-1}(\overline{S})$ of $\Phi^{-1}_{T_n}(B(\overline{T_n}))$ may correspond to infinitely many triangles of a single isometry class of triangles.  \\

\textbf{Claim.} $A_n$ is an acute triangulation of $T$.

\begin{proof}[Proof of the Claim]
The only triangles of $\Phi^{-1}_{T_n}(B(\overline{T_n}))$ affected by this last merging step are triangles with at least one vertex lying on an edge of $T_n$. There are 3 possibilities, depending on the number of vertices lying on an edge of $T_n$ (1, 2 or 3). However, by construction, $B(\overline{T_N})$ does not have triangles with all 3 vertices lying on an edge of $\overline{T_N}$, and thus neither does $\Phi^{-1}_{T_n}(B(\overline{T_n}))$ (see \cite{Bishop}). Both of the remaining cases (see Fig. \ref{fig:cases}) will be resolved with the same treatment.

\begin{figure}[H]\centering
  \includegraphics[page=6]{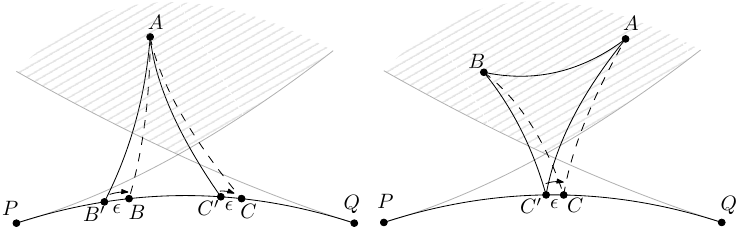}
  \caption{: The two possible types of problematic triangles: (Left) triangles with two vertices lying on an edge of $T_n$ and (Right) triangles with one vertex lying on an edge of $T_n$.}
  \label{fig:cases}
\end{figure}

Let $ABC$ be such an affected triangle of $\Phi^{-1}_{T_n}(B(\overline{T_n}))$, as seen after the merging step. Let us assume that $ABC$ lies inside a triangle $PQR$ of $T_n$, and suppose without loss of generality that $PQ$ is the edge supporting either one or two vertices of $ABC$ (which we will assume to be $A$ and/or $B$). We denote by $a$, $b$, and $c$ the lengths of the side of $ABC$ opposite $A$, $B$ and $C$; and by $\alpha,\beta,\gamma$ the angles based at the same vertices. For all $i\in [M]$, we shall write $\kappa_i$ to denote the curvature of the surface of constant curvature supporting the triangle $t_i$. Fix then $i\in[M]$ and let us write $A'$, $B'$ and $C'$ to denote the positions the vertices $A$, $B$ and $C$ occupied in the triangle $PQR$ before the merging operation associated with the triangle $t_i$; and the associated angles by $\alpha'$, $\beta'$ and $\gamma'$. Recall that two quantities $x$ and $y$, both depending on a given parameter, are called \textit{comparable} if there exists a constant $L<\infty$ such that, for all value of the parameter, we simultaneously have $x<Ly$ and $y<Lx$. In our context, the parameter will be understood to be the number of steps $n>N$ in the iterated medial subdivision of $T$.\\

\textbf{Fact.} The edge lengths $a'$, $b'$, $c'$ and $d_{M_{\kappa_i}^2}(P,Q)$ are all comparable to one another.
\begin{proof}[Proof of the fact]
The uniform angle bound given by the first implication $(i)$ of Lemma E, combined with the (spherical or hyperbolic) law of sines (see \cite{WT}) shows that the quantities $a'$, $b'$ and $c'$ are comparable. For the same reason, all side-lengths of the triangle $PQR$ are comparable. Additionally, the second implication $(ii)$ of Lemma E tells us that $b'$ (for example) is comparable from below to $d_{M_{\kappa_i}^2}(P,Q)$. Indeed, thanks to Lemma E $(ii)$, the (spherical or hyperbolic) tangent formula (\cite{Carslaw}, p. 100) implies a lower bound on $b'$ in terms of $d_{M_{\kappa_i}^2}(Q,R)/2$ and $\theta$. On the other hand, $PQR$ is a convex geodesic triangle and therefore $b'$ is upper-bounded by the largest side-length of $PQR$, as $A'B'C'$ lies inside $PQR$. This largest side-length is either $d_{M_{\kappa_i}^2}(P,Q)$ itself or is comparable to it by our previous observation. Note that this treatment is valid in both cases since the existence of a single vertex subject to Lemma E $(ii)$ is enough.
\end{proof}

Let us write:

\[
\epsilon(\overline{S}) \coloneqq \max\{d_{M_{\kappa_i}^2}(\phi_{t_i}^{-1}(\overline{S}),S) \> | \> i \in [M]\}
\]
and:
\[
\epsilon\coloneqq \max \{\epsilon(\overline{S}) \> | \> \overline{S}\in E(\overline{T_n})\}
\]

where $E(\overline{T_n})$ denotes the set of open edges of $\overline{T_n}$. Note that when $\bar{S}$ is the endpoint of an edge of $\overline{T_n}$, all its preimages under $\phi_{t_1}$, $\phi_{t_2}$, \ldots, $\phi_{t_M}$ coincide with $S$ at a vertex of $T_n$, so that the function $\epsilon:E(\overline{T_n}) \rightarrow \mathbb R$ has compact image and $\epsilon$ is well defined. By definition, we have that, for all $i\in[M]$, $d_{M_{\kappa_i}^2}(A,A')<\epsilon$, $d_{M_{\kappa_i}^2}(B,B')<\epsilon$ and $d_{M_{\kappa_i}^2}(C,C')<\epsilon$. Recalling the definition of $S$, we can express $\epsilon$ as the following quantity:

\[
\epsilon = \max_{i\in [M]}\left| d_{M_{\kappa_i}^2}(P,B') \left( \frac{|\overline{P}\>\overline{B'}|}{d_{M_{\kappa_i}^2}(P,B')} - 1 \right) \right| <   \max_{i\in [M]}d_{M_{\kappa_i}^2}(P,Q) \cdot \max_{i\in [M]}\left| \frac{|\overline{P}\>\overline{B'}|}{d_{M_{\kappa_i}^2}(P,B')} - 1 \right|
\]

The previous fact then implies the following inequality:

\[
\epsilon < \max_{i\in [M]}\{L^i_a,L^i_b,L^i_c\} \cdot \max\{a,b,c\}  \cdot\max_{i\in [M]}\left| \frac{|\overline{P}\>\overline{B'}|}{d_{M_{\kappa_i}^2}(P,B')} - 1  \right|
\]

where $L^i_a$, $L^i_b$ and $L^i_c$ are the comparability constants of $a'$, $b'$, $c'$ with respect to $d_{M_{\kappa_i}^2}(P,Q)$.

By Claim 2 of Lemma D and Lemma A, we know that the expressions within the absolute value becomes arbitrarily close to 0 as $n$ gets arbitrarily large. This shows that, for $n$ large enough, $\epsilon$ can be chosen to be as small as required with respect to $a'$, $b'$, and $c'$. The hyperbolic and spherical cosine law (\cite{WT}) then implies that, for $n$ large enough, the value of $\cos \alpha$ can be taken arbitrarily close to that of $\cos \alpha'$ (and likewise for that of $\cos \beta$ and $\cos \gamma$).

% \[
% \cos \alpha = \frac{\cosh b \cosh c-\cosh a}{\sinh b\sinh c}
% \]
\end{proof}

\end{proof}

\section{Final Remarks and Future Work: Towards Acute Triangulations of Riemannian Polygonal Complexes}
\label{5}
We first point out to the reader that our proof of acuteness ultimately relies on a compactness argument of Bishop (\cite{Bishop}) and therefore does not provide explicit angle bounds. Furthermore, a uniform bound is also lost, as we are relying on the value of $\delta$ given by Lemma B. In \cite{Flo}, it is shown that this $\delta$ may be arbitrarily small depending on the original size of the complex under consideration. 

Lastly, we remark that Beltrami's theorem (see \cite{Carmo} for example) provides an obstruction to a direct extension of our method to more general Riemannian metrics. Indeed, any metric which is locally projectively equivalent to the Euclidean metric is necessarily a constant curvature metric. Nevertheless, an extension to general Riemannian metrics should be straightforward by considering small enough patches where the curvature is near-constant. In upcoming work, we intend to bootstrap the results of \cite{Flo} and the present article to recover Lemmas A, B and D in the general Riemannian setting.

\section{Acknowledgements}
\label{sec6}

I would like to express my deepest gratitude to Piotr Przytycki, for his \mbox{unwavering} support and guidance; and to Christopher Bishop, for the clarity of his work in \cite{Bishop} and his generous help. I also thank Mathijs Wintraecken for his useful remarks and comments.

%==============================================================

\end{document}